\documentclass[twocolumn]{article}

\def\BibTeX{{\rm B\kern-.05em{\sc i\kern-.025em b}\kern-.08em
    T\kern-.1667em\lower.7ex\hbox{E}\kern-.125emX}}

\usepackage{booktabs} 
\usepackage{paralist} 
\usepackage{multirow}
\usepackage{bbding}
\usepackage{pifont}
\usepackage{wasysym}
\usepackage{amssymb}
\usepackage{xcolor}
\usepackage{color, colortbl}
\usepackage{tikz}
\usepackage{subcaption}
\usepackage{enumitem}
\usepackage[skip=0pt]{caption}
\usepackage{tabularx}
\usepackage{amsthm}
\usepackage{url}

\newcommand{\ex}{prov-example}

\usetikzlibrary{calc,positioning,shapes.multipart,decorations.pathreplacing,shapes.arrows,decorations.pathreplacing, shapes.misc}

\definecolor{red1}{RGB}{195,0,0}
\definecolor{red2}{RGB}{246,136,93}
\definecolor{yellow1}{RGB}{247,175,47}
\definecolor{yellow2}{RGB}{255,192,96}
\definecolor{yellow3}{RGB}{255,255,96}
\definecolor{green1}{RGB}{214,249,121}
\definecolor{green2}{RGB}{113,158,65}

\newtheorem{theorem}{Theorem}[section]
\newtheorem{proposition}[theorem]{Proposition}
\newtheorem{example}[theorem]{Example}

\newtheorem{definition}[theorem]{Definition}

\def\TextShift{1pt}

\tikzset{
  myrect/.style={
    rectangle split, 
    rectangle split horizontal,
    rectangle split parts=#1,
    anchor=west,
  },
  mytext/.style={
    arrow box,
    draw=#1!70!black,
    fill=#1,
    align=center,
    line width=1pt,
    font=\sffamily,
    font=\small
  },
  mytextb/.style={
    mytext=#1,
    anchor=north,
    arrow box arrows={north:0.5cm}  
  },
  mytexta/.style={
    mytext=#1,
    anchor=south,
    arrow box arrows={south:0.35cm}  
  }
}

\newcommand\AddText[5][]{
  \if#5l\relax
    \node[mytextb=#2,yshift=-\TextShift,#1] 
      at (part#4.south west) {\strut#3\strut};
  \fi
  \if#5L\relax
    \node[mytexta=#2,yshift=\TextShift,#1] 
      at (part#4.north west) {\strut#3\strut};
  \fi
  \if#5m\relax
    \node[mytextb=#2,yshift=-\TextShift,#1] 
      at ( $ (part#4.south west)!0.5!(part#4.south east) $ ) {\strut#3\strut};
  \fi
  \if#5M\relax
    \node[mytexta=#2,yshift=\TextShift,#1] 
      at ( $ (part#4.north west)!0.5!(part#4.north east) $ ) {\strut#3\strut};
  \fi
  \if#5r\relax
    \node[mytextb=#2,yshift=-\TextShift,#1] 
      at (part#4.south east) {\strut#3\strut};
  \fi
  \if#5R\relax
    \node[mytexta=#2,yshift=\TextShift,#1] 
      at (part#4.north east) {\strut#3\strut};
  \fi
}

\AtBeginDocument{%
  \providecommand\BibTeX{{%
    \normalfont B\kern-0.5em{\scshape i\kern-0.25em b}\kern-0.8em\TeX}}}

\date{}

\begin{document}

\title{Towards Inferring Queries from Simple and Partial Provenance Examples}

\author{
  Amir Gilad\\ Tel Aviv University \\ amirgilad@mail.tau.ac.il
  \and
  Yuval Moskovitch \\ University of Michigan \\ yuvalm@umich.edu
}


\maketitle

\begin{abstract}
The field of query-by-example aims at inferring queries from output examples given by non-expert users, by finding the underlying logic that binds the examples.
However, for a very small set of examples, it is difficult to correctly infer such logic. To bridge this gap, previous work suggested attaching {\em explanations} to each output example, modeled as provenance, allowing users to explain the reason behind their choice of example. In this paper, we explore the problem of inferring queries from a few output examples and {\em intuitive explanations}. 
We propose a two step framework: (1) convert the  explanations into (partial) provenance and (2) infer a query that generates the output examples using a novel algorithm that employs a graph based approach. This framework is suitable for non-experts as it does not require the specification of the provenance in its entirety or an understanding of its structure. We show promising initial experimental results of our approach. 
\end{abstract}

\section{introduction}
With the growing interest in data science and the ubiquity of data in recent years, the need for tools that allow non-expert users to interact with databases becomes crucial.
Different approaches have been proposed to facilitate manners in which non-expert users can query the database, including natural language interfaces~\cite{nalir} and query by example \cite{shen,Bonifati}, where the user provides the system output examples in order to obtain a query that produces the same output. 
In this paper we focus on an alternative approach, based on the use of explanations for output examples. This notion relies on the perception that providing a small number of output examples, along with explanations may be more intuitive and easier for users than solely providing many output examples. These explanations encapsulate information on the query structure which significantly reduces the number of examples needed to infer the intended query.

\begin{example}\label{ex:running}
Consider the Microsoft Academic Database (MAS)~\cite{mas}. 
An instance with a similar schema is depicted in Figure~\ref{fig:db}. Also consider the query returning all conferences in the database area and authors from Tel Aviv University (TAU) who published papers in these conferences (shown as a Conjunctive Query in Figure \ref{fig:sql}). A user who does not know how to formulate it can input a few examples of outputs and explain their rational. For instance, the two output examples can be ``SIGMOD, Alice'' and ``CIKM, Bob'', and their respective explanations saying that Alice published the papers ``X'' and ``Y'' in SIGMOD, which is a database conference, and that Bob published the paper  ``Z'' in CIKM which is also a database conference, and both are from Tel Aviv University. 
\end{example}

The idea of inferring queries from provenance was first proposed in~ \cite{DeutchG16,DeutchG19}. 
A key property of any provenance based solution is its level of detail. In \cite{DeutchG19} we studied provenance polynomials, trio, positive Boolean expressions, and why-provenance, each corresponds to an explanation with a different granularity. 
We showed that consistent conjunctive queries can be inferred, even when given only a few output examples and their provenance. A notable provenance model that is missing in \cite{DeutchG19} is that of lineage~\cite{lin} where the explanation of an output tuple $t$ is the set of input tuples' annotations that contribute to the generation of $t$. 

\begin{example}\label{ex:provenance}
Consider again the database depicted in Figure \ref{fig:db} and the query from Example \ref{ex:running}. The lineage of the tuple ``SIGMOD, Alice'' is the set of tuples annotated (in the prov. column) by $o_2, a_2, p_1, p_2, w_1, w_2, c_2, dc_2$ and $d_1$. Note that this set depicts two different ways to obtain the output, one using the fact that she published ``X'' and the other that she published ``Y''. 
\end{example}



Providing precise and ``full'' provenance as explanations may be tedious or even impossible for non-expert users, and requires an understanding of the schema. A more likely non-expert explanation would be partial. For instance, a more natural explanation for the output  ``SIGMOD, Alice'' from Example \ref{ex:running} may contain only the tuples $o_2,a_2, p_1, c_2$ and $d_1$. Moreover, a major barrier in the usability of the provenance based query inference approach, is that formulating explanations into provenance, even in a partial form, may be a challenging task for non-expert users. To this end we propose a two phase framework. Given a small set of output examples along with informal explanations, we first  convert these explanations into (partial) provenance. Then, we infer a query that generates each output example from its corresponding provenance.


{\em In this paper we describe our ongoing work toward a solution for the problem of inferring queries based on examples with intuitive explanations.}
We propose to use values as explanation, and leverage techniques developed in previous work~\cite{Psallidas} for the first part. The basic idea is to identify the tuples that contain the values given as explanations using a similarity score. 


\begin{example}\label{ex:values}
Consider again the database shown in Figure \ref{fig:db} and the output example 'SIGMOD, Alice'. A possible explanation may contain values such as `Tel Aviv University', `Alice', ``X'', ``Y'', `SIGMOD', and `Databases'. These values would then be mapped to the tuples $o_2, a_2, p_1, p_2, c_2, d_1$, respectively. 
\end{example}

The second part of the proposed solution is inferring queries from output examples and partial lineage.
An essential question that arises within this notion of partial provenance, is {\bf how partial} is the provenance. Possible characterization of the explanation's form may include the presence or absence of projected tuples or join tuples. 
In this paper, we describe our results for \emph{joinless lineage}--e.g, lineage with missing join tuples. We start with this fragment since it allows for intuitive explanations, and it overcomes an inherent limitation of the mapping component where some tuples in the lineage are absent in the mapping values function range. For instance, there is no value in Example \ref{ex:values} that can be mapped to the tuples $w_1$ and $w_2$, although it is part of the provenance of the intended query, as shown in Example \ref{ex:provenance}. 
Other possible forms of partial lineage are left for our ongoing and future work.
\begin{figure}
    \centering \scriptsize
    \begin{minipage}{.5\linewidth}
        \centering
        \caption*{\footnotesize Rel. \textit{org}}\label{tbl:organization}
        \begin{tabular}{ c | c | c | c | c | c |}
            \cline{2-3} & oid & oname \\
            \hline $o_1$ & 1 & UMICH \\
            \hline $o_2$ & 2 & TAU \\
            \hline
        \end{tabular}
        
    \end{minipage}%
    \begin{minipage}{.5\linewidth}
        \centering
        \caption*{\footnotesize Rel. \textit{author}}\label{tbl:author}
        \begin{tabular}{ c | c | c | c | c | c |}
            \cline{2-4} & aid & aname & oid \\
            \hline $a_1$ & 3 & Carol & 1 \\
            \hline $a_2$ & 4 & Alice & 2 \\
            \hline $a_3$ & 5 & Bob & 2 \\
            \hline
        \end{tabular}
        
    \end{minipage}
    \begin{minipage}{.4\linewidth}
        \centering
        \caption*{\footnotesize Rel. \textit{domain\_conf}}\label{tbl:domainconference}
        \begin{tabular}{ c | c | c | c | c | c |}
            \cline{2-3} & cid & did\\
            \hline $dc_1$ & 10 & 18 \\
            \hline $dc_2$ & 11 & 18 \\
            \hline
        \end{tabular}
        
    \end{minipage}%
    \begin{minipage}{0.6\linewidth}
        \centering
        \caption*{\footnotesize Rel. \textit{pub}}\label{tbl:publication}
        \begin{tabular}{ c | c | c | c | c | c |}
            \cline{2-5} & wid & cid & ptitle & pyear\\
            \hline $p_1$ & 6 & 11 & ``X..." & 2014 \\
            \hline $p_2$ & 7 & 11 & ``Y..." & 2014 \\
            \hline $p_3$ & 8 & 10 & ``Z..." & 2007 \\
            \hline
        \end{tabular}
        
    \end{minipage}
    \begin{minipage}{.5\linewidth}
        \centering
        \caption*{\footnotesize Rel. \textit{writes}}\label{tbl:writes}
        \begin{tabular}{ c | c | c | c | c | c |}
            \cline{2-3} & aid & wid\\
            \hline $w_1$ & 4 & 6 \\
            \hline $w_2$ & 4 & 7 \\
            \hline $w_3$ & 5 & 8 \\
            \hline $w_4$ & 3 & 6 \\
            \hline
        \end{tabular}
        
    \end{minipage}%
    \begin{minipage}{.5\linewidth}
        \centering
        \caption*{\footnotesize Rel. \textit{conf}}\label{tbl:conference}
        \begin{tabular}{ c | c | c | c | c | c |}
            \cline{2-3} & cid & cname\\
            \hline $c_1$ & 10 & CIKM \\
            \hline $c_2$ & 11 & SIGMOD \\
            \hline
        \end{tabular}
        
    \end{minipage}
        
    \begin{minipage}{.33\linewidth}
        \centering
        \caption*{\footnotesize Rel. \textit{domain}}\label{tbl:domain}
        \begin{tabular}{ c | c | c | c | c | c |}
            \cline{2-3} & did & dname\\
            \hline $d_1$ & 18 & DB \\
            \hline
        \end{tabular}
        
    \end{minipage}

    \caption{DB instance}\label{fig:db}
    \vspace{-5mm}
\end{figure}

\paragraph*{Related Work}
Query-by-example \cite{BonifatiCS16,shen,Psallidas}, is the problem of inferring queries based on output examples given by the user. 
Along side these works, various models of provenance have been proposed in the literature~\cite{GKT-pods07,lin,Why}. 
Recent work~\cite{DeutchG16,DeutchG19} has explored the case where a small number of positive examples are provided, coupled with explanations for these examples in the form of provenance. These assume that outputs with multiple explanations are separated into individual ones and either are given in full, or are missing mention of duplicate uses of the same tuple in an explanation. 
Our approach attempts to fill in the gap focusing on the lineage formalism \cite{lin}, which was not included in \cite{DeutchG19}, and thus relaxes the desiderata of the explanations, allowing for missing tuples and unseparated explanations for the same output. 

\section{Initial Results}
\label{sec:queryInf}


We first describe our initial results for the query inference phase. We formally define the problem of query inference from partial lineage and present a novel algorithm for the joinless lineage case. We then briefly discuss our approach for provenance extraction.


\subsection{The Query Inference Problem}
We use the foundations laid in \cite{DeutchG19} to define the problem of inferring queries by lineage. As in \cite{DeutchG19}, we focus on Conjunctive Queries (CQ). 
An assignment of a CQ $Q$ with respect to a database $D$ is a mapping of the relational atoms of $Q$ to tuples in $D$ that respects
relation names.
The input to the problem is a set of output tuples, each with its (possibly partial) lineage. The lineage is given as a set of input tuples identifiers. We use $(I,O)$ to denote a pair of output example $O$ and it's lineage $I$ (i.e., a set of input tuple identifiers).  As a simple example, consider the table in Figure \ref{fig:example}, referring to the annotation of the tuples in Figure \ref{fig:db}. The output column along with the Full Prov. column forms an example where the first and seconds rows can be denoted by $(I_1, O_1)$ and $(I_2, O_2)$. The Partial Prov. column depicts examples of partial lineage for the same outputs.

\setlength{\belowcaptionskip}{-5pt}
\begin{figure}
    
    \centering
        \scriptsize{
            \begin{tabular}{| c | c | c | c | c | c |}
            \hline Outupt & Full Prov. & Partial Prov. \\
            \hline SIGMOD, Alice & \begin{tabular}{@{}c@{}}$o_2, a_2, p_1, p_2, w_1,$\\$w_2, c_2, dc_2, d_1$\end{tabular} & $o_2, a_2, p_1, c_2, d_1$ \\
            \hline EDBT, Bob & \begin{tabular}{@{}c@{}}$o_2, a_3, p_3, w_3,$\\ $c_1, dc_1, d_1$\end{tabular} & $o_2, a_3, p_3, c_1, d_1$ \\
            \hline
        \end{tabular}
        }
        \caption{\ex}\label{fig:example}
    \end{figure}
    
    \begin{figure}
        \centering
        \scriptsize{
        \begin{tabular}{|l|}
            \hline
            \verb"q(cname, aname) :- author(aid, aname, oid),"\\
            \verb"writes(aid, wid), pub(wid, cid, ptitle, pyear),"\\
            \verb"conf(cid, cname), domain_conf(cid, did),"\\
            \verb"domain(did, dname), oname = `TAU', dname = `Databases'"\\
            \hline
        \end{tabular}}
        \caption{Intended query}\label{fig:sql}
    \vspace{-2mm}
    \end{figure}

We can then define the notion of consistency with respect to an output example and (partial) lineage, and introduce our problem statement. Intuitively, we look for a query whose output contains the example output, and for each output tuple, its provenance is ``reflected'' in the computation of the tuple by the query. 

\begin{definition}[adapted from \cite{DeutchG19}]\label{def:gen_problem}
Given an $(I,O)$, a database $D$, and a CQ $Q$, we say that $Q$ is consistent with respect to the example if there exists $S\subseteq D$ such that
$O\in Q(I \cup S)$ and $I\subseteq lin(Q|_O(I\cup S))$, where $lin(Q|_O(I\cup S))$ is the lineage of the tuple $O$ according $Q$ evaluated on $I\cup S$. 
\end{definition}

The provenance specified in the explanation can be partial but has to be a part of the assignment and cannot includes irrelevant tuples. A partial explanation is any non-empty subset of the lineage of the output.
For example, consider the the partial explanation, $I$, shown in the Partial Prov. column in the first row in Figure \ref{fig:example}. The query in Figure \ref{fig:sql} is consistent w.r.t. it since we have $S = \{w_1, dc_2\}$ such that $(SIGMOD, Alice) \in Q(I\cup S)$ and $I$ is a partial explanation.



We refer to a set of pairs $(I,O)$ as a {\em \ex}. Given a CQ $Q$, we say that $Q$ is consistent with respect to this \ex\ if it is consistent with each one of the $(I,O)$ pairs it contains. A consistent query can be very general, in fact we showed in~\cite{DeutchG19} that there exists a set of $(I,O)$ examples with an infinite number of non-equivalent consistent queries. 
A major factor influencing the number of consistent queries is the length of the query, and in particular, the number of self-joins allowed. 
Therefore, a natural desideratum is a small number of joins. 
We utilize the concept of {\em join graph} for conjunctive queries \cite{Kalashnikov2018}. 
In a join graph every two atoms joined together in the CQ are joined by an edge .
For instance, the join graph of the query in Figure \ref{fig:sql} is depicted in Figure \ref{fig:join}. We say that a query is connected if its join graph is connected.




\setlength{\belowcaptionskip}{-2pt}
\begin{figure}
\centering
\begin{subfigure}[b]{.5\linewidth}
\centering
\includegraphics[scale = 0.6]{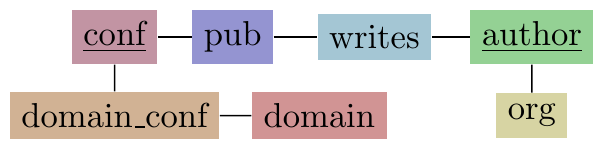}
\caption{Join graph of the query} \label{fig:join}
\end{subfigure}%
\begin{subfigure}[b]{.5\linewidth}
\centering
\includegraphics[scale = 0.6]{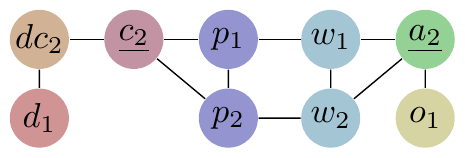}
\caption{Prov. graph of row 1}\label{fig:prov_graph}
\end{subfigure}%

\caption{(a) Query join graph in Figure \ref{fig:sql}, (b) provenance graph of row 1 in Figure \ref{fig:example}} 
\end{figure}

\begin{definition}\label{def:minimal}

A consistent query $Q$ with respect to a given \ex\  is consistent-minimal if (1) $Q$ is consistent and connected, and (2) for every query $Q'$ such that $Q'$ is consistent and connected, the number of nodes in the join graph of $Q$ is smaller or equal to the number of nodes in the join graph of $Q'$. 
\end{definition}

Intuitively, consistent-minimal queries are more natural as they offer a concise reason for the given output according to the explanation. 
To demonstrate, consider the following query \\
\begin{footnotesize}
\begin{tabular}{l}
\verb"q(cname, aname) :- author(aid, aname, oid),"\\
\verb"writes(aid, wid1), pub(wid1, cid, ptitle1, pyear1),"\\
\verb"writes(aid, wid2), pub(wid2, cid, ptitle2, pyear2),"\\
\verb"conf(cid, cname), domain_conf(cid, did),"\\
\verb"domain(did, dname), oname = `TAU', dname = `Databases'"\\
    \end{tabular}   
\end{footnotesize}

\noindent
It is consistent but not minimal. Conversely, the query depicted in Figure \ref{fig:sql} is consistent-minimal. 
Our goal is then to find a consistent-minimal query given a \ex.





\subsection{Provenance Graph}\label{subsec:prov_graph}
Given a pair $(I,O)$ where $I = \{t_1, \ldots, t_n\}$, its provenance graph $G_P = (V_P, E_P)$ is defined as follows. $V_P$ are the tuple annotations and $E_P = \{\{t_i,t_j\}\mid \exists A,B.~ t_i.A = t_j.B\}$, i.e., there is an edge between each two tuples who share a constant, including self edges (e.g., the tuple $R(1,1)$ would have a self edge). 
The provenance graph of the first row of the \ex\ in Figure~\ref{fig:example} is shown in Figure~\ref{fig:prov_graph}. 
We next establish the connection between the join and provenance graphs through the concept of graph homomorphism.

\begin{definition}\label{def:hom}
Let $G_J= (V_J, E_J)$ be a query join graph and $G_P = (V_P, E_P)$ a provenance graph of an $(I,O)$. A graph homomorphism between $G_J$ and $G_P$ is a function $h:V_J\to V_P$ such that (1) $(h(u),h(v))\in E_P$ if $(u,v)\in E_J$, particularly, for every variable shared by $u,v$ at index $i$, $h(u),h(v)$ share a constant at the same index,
(2) $u, h(u)$ have the same relation name, and (3) if $v\in V_J$ has a projected variable at index $i$ to the head of $Q$ at index $j$, then $h(v)$ has the same constant at index $i$ as the constant of $O$ at index $j$. 
\end{definition}

\begin{example}\label{ex:homo}
Reconsider the join graph $G_J$ in Figure \ref{fig:join} and the provenance graph $G_P$ in Figure \ref{fig:prov_graph}. An homomorphism $h:V_J \to V_P$ is $h(conf) = c_2$, $h(author) = a_2$, $h(pub) = p_1$, $h(writes) = w_1$, $h(org) = o_2$, $h(domain) = d_1$, $h(domain\_conf) = dc_2$. 
\end{example}


Given a set of homomorphisms $\mathcal{H} = \{h_1, \ldots, h_k\}$ between a join graph $G_J$ and a provenance graph $G_P$, we say that $\mathcal{H}$ {\em covers} $G_P$ if $\bigcup_{h \in \mathcal{H}} h[V_J] = V_P$. 
Let $Q$ be a CQ with a join graph $G_J$, let $(O,I)$ be a row of a \ex\ with a provenance graph $G_P$. We can show that.

\begin{proposition}\label{prop:consistent_graph}
There exists a set $\mathcal{H}$ of homomorphisms from $G_J$ to $G_P$ that covers $G_P$ iff $Q$ is consistent with respect to $(O,I)$. 
\end{proposition}

\begin{proof}
Suppose we have a set of homomorphisms $\mathcal{H}$ that covers $G_P$, so every $h \in \mathcal{H}$ maps all the nodes of $G_J$ to some of the nodes of $G_P$ such that Definition \ref{def:hom} holds. It is enough to show that every $h \in \mathcal{H}$ is equivalent to an assignment of the provenance tuples to $Q$ that produces the otuput tuple, since $\bigcup_{h \in \mathcal{H}} h[V_J] = V_P$ so all provenance tuples are used for some assignment to $Q$. Condition (1) in Definition \ref{def:hom} defines that if two atoms share a variable, the two tuples to which they are mapped to also share a constant in the same index. Condition (2) says that all homomorphisms map atoms to tuples with the same relation name, and finally condition (3) says that the mapping produces the output tuple. Thus, every $h \in \mathcal{H}$ is essentially an assignment to $Q$.
Similarly for the other direction, assume $Q$ is consistent, then it uses all of the provenance tuples to generate the output tuple. Every assignment defines a homomorphism that maps the join graph $G_J$ to the provenance graph $G_P$ such that Definition \ref{def:hom} holds. Denote these homomorphisms by $\mathcal{H}$.
Since $Q$ uses all provenance tuples, $\mathcal{H}$ covers $G_P$. 
\end{proof}

For instance, the query in Figure \ref{fig:sql} is consistent w.r.t. the \ex\ in Figure \ref{fig:example} since there are two homomorphisms that map its join graph in Figure \ref{fig:join} to the provenance graph of the first row in Figure \ref{fig:prov_graph}. One is described in Example \ref{ex:homo}. 


From Proposition \ref{prop:consistent_graph} it follows that given a pair $(I,O)$ and a query $Q$ such that (1) $Q$ has the smallest connected join graph and (2) there is a set of homomorphisms from the join graph of $Q$ to the graph of $(I,O)$ that covers it, then $Q$ is consistent-minimal. We use this observation when inferring queries as we show next. 

\subsection{Homomorphism Based Algorithm}
We start by presenting a solution for the case where a \ex\ includes full lineage, and then relax this solution to account for joinless lineage. The input to the algorithm is a \ex\ $Ex$ and the database schema graph, and it consists two parts as follows. 
The first step of the algorithm is finding the projected attributes of the query. To this end, we generate a list of candidates from each row of $Ex$. The candidates $C_i$ of the row $(O_i, I_i)$ are the attribute of the subset of relation names of tuples in $I_i$ that share a common constant with $O_i$. The algorithm then intersect the candidates of the projection to get the projected attributes $\bigcap_i C_i$.

\begin{example}
Reconsider the first row of the \ex\ shown in Figure \ref{fig:example}. The output values `SIGMOD' and `Alice' are searched in the provenance tuples. Since the tuples annotated by $c_2, a_2$ are the only ones that contain them, they are the candidates from the first row. In the same process, we get the tuples $c_1, a_3$ from the second row. For each output attribute, we intersect the attributes that are candidate for projection, e.g., $\{conf.cname\} \cap \{conf.cname\} = \{conf.cname\}$ for the first attribute. 
\end{example}


At a high level, the main idea of the second step of the algorithm is to generate all possible join graphs, and for each such graph, search for a set of homomorphisms that covers the provenance graph of each one of the rows in $Ex$.  When such a set is found the corresponding query is returned.
The join graphs generation is done in
a growing size order, which guarantees that the first graph satisfying the coverage condition is minimal.
In fact, as  shown in~\cite{DeutchG19}, if a consistent query exists, it must be of length at least $d$ and at most $k+d\cdot n$, where $k$ is the number of attributes of the output tuple, $d$ is the number of distinct relations in the provenance, and $n$ is the number of tuples in the largest provenance set. Thus, we consider only graphs with a least $d$ nodes and at most $k+d\cdot n$ that satisfy essential conditions for a consistent-minimal query: (1) they are connected, and (2)~they include exactly the same relation names appearing in the provenance of each one of the rows in $Ex$. 
Finally, the algorithm adds selection conditions, by
inspecting possible values for constants in each attribute according to the tuples in the \ex. For attributes that take a single value, it assigns a constant in the query.

\begin{example}
Reconsider the \ex\ in Figure \ref{fig:example} with the Full Prov. column and the schema of the database in Figure~\ref{fig:db}. The provenance graph of the first row is depicted in Figure~\ref{fig:prov_graph}. The algorithm finds the candidates for projection which are $conf$.$cname$ and $author$.$aname$ in both rows. We then start generating join graphs. We consider only connected graphs, that include all the relations appearing in the \ex\ with at least $7$ nodes, as the number of distinct relations in the \ex. The only graph that fits this description is given in Figure \ref{fig:join}.  As shown in Section \ref{subsec:prov_graph}, there is a set of homomoprphisms that covers the provenance graph of each one of the rows in the \ex. 
Finally, the algorithm adds constants and returns the query in Figure~\ref{fig:sql}.
\end{example}


\paragraph*{Joinless lineage} 
We next present a heuristic to account for \emph{joinless lineage}, a fragment of partial lineage where join relations containing only foreign keys is missing from the lineage. 
More formally, suppose $\{t_1, t_2, t_3, \ldots, t_k\} \subseteq I$ constitute a full assignment present in the provenance set of an output tuple $O$. Assume the relation of $t_1$ is joined with the relation of $t_3$ in the schema by a join table $t_2$ such that the attributes of the table of $t_2$ are only the primary keys of $t_1$ and $t_3$, then $I \setminus \{t_2\}$ is an acceptable explanation.
This fragment is particularly appealing since it allows for intuitive explanations while nicely preserving the information encapsulated in the lineage. 
In this case, the algorithm will ``complete'' the explanation by finding the missing tuples. We are given a \ex\ where the provenance $I$ in each row contains all tuples except join relations that connect two tuples $t_i, t_j$ that are present in $I$.
We add a procedure after finding the projection attributes, that checks whether the provenance graph $G_P$ of each row is connected. If not, there is a gap of size one between two tuples $t_i$ and $t_j$ (this can be detected by finding the shortest path between them in the schema graph). We then look for a relevant tuple to connect $t_i$ and $t_j$ (with relations $R_i$ and $R_j$ resp.) through another relation $R_v$. This is done by querying the database to get the tuple $t_{v}$.

\begin{example}
Consider the \ex\ in Figure \ref{fig:example} with the Partial Prov. column.
We find that the provenance graph is not connected, and specifically, that $a_2$ and $p_1$ do not share a constant and their relations, according to the schema, are connected by the $writes$ table, so we query the $writes$ relation in Figure \ref{fig:db}, with the $aid$ of $a_2$ and the $wid$ of $p_1$. The result is the tuple $w_1$ so we add it to the provenance. We do the same for $c_2$ and $d_1$ to find $dc_2$. In the second row, we add the tuples $w_3$ and $dc_1$ and then continue to the graph generation step to find the query.
\end{example}
\subsection{From Values to Provenance Tuples}\label{sec:convert}



To allow users to easily formulate explanations, we assume that the initial explanation is composed of individual values, and the values are then mapped to tuples in the database. The challenge is finding tuples  that most resemble the intended tuples. We will draw on works such as \cite{Psallidas} and use a value similarity score that compares the values inputted by the user and the tuples in the database. The score of a tuple $t$ w.r.t an inputted value $v$ can be defined as $score(t,v) = max_{t[a]} sim(t[a],v)$, i.e., the maximum similarity between a cell value in $t$, $t[a]$, and the value~$v$. The $sim$ function can be changed according to the context, e.g., string similarity or $L_1$ norm for numbers. 
In our implementation, we have used a greedy approach so that if a cell contains a string that exactly matches the inputted string, we return it immediately. 

\begin{example}
Reconsider our running example with the output examples depicted in Figure \ref{fig:example}. An explanation for Alice, SIGMOD may be `TAU', `Alice', `X', `Y',  `SIGMOD', `Databases'. The values would then be mapped the tuples $o_2, a_2, p_1, p_2, c_2, d_1$ since these are the tuples most similar to the inputted values as they have an attribute that is identical to them. For example, $a_2[aname] = Alice$. 
\end{example}



\section{Experiments}\label{sec:exp}

        
        
        

\begin{table}[!ht]
    \centering \footnotesize
    \caption{Queries used in the experiments}\label{tbl:queries}
    \begin{tabularx}{\linewidth}{| c | X | c | c | c | c |}
        \hline {\bf Num.} & {\bf Query} \\
        \hline 1 & SELECT author.name FROM writes, pub, conf, author WHERE writes.pid = pub.pid AND writes.aid = author.aid AND pub.cid = conf.cid AND conf.name LIKE 'SIGMOD'; \\
        
        \hline 2 & SELECT pub.title FROM conf, domain\_conf, pub, domain\_pub, domain WHERE conf.cid = domain\_conf.cid AND conf.cid = pub.cid AND domain\_conf.did = domain.did AND pub.pid = domain\_pub.pid AND domain.name LIKE 'Databases'; \\
        
        \hline 3 & SELECT author.name FROM writes, pub, conf, author WHERE writes.pid = pub.pid AND writes.aid = author.aid AND pub.cid = conf.cid AND conf.name LIKE 'SIGMOD' AND pub.year > 2005; \\
        
        \hline 4 & SELECT author.name FROM writes, pub, conf, author WHERE writes.pid = pub.pid AND writes.aid = author.aid AND pub.cid = conf.cid AND conf.name LIKE 'SIGMOD' AND pub.year > 2005 AND pub.year < 2015; \\
        
        \hline 5 & SELECT author.name FROM author, writes, conf, domain\_conf, pub, domain\_pub, domain WHERE author.aid = writes.aid AND writes.pid = pub.pid AND conf.cid = domain\_conf.cid AND conf.cid = pub.cid AND domain\_conf.did = domain.did AND pub.pid = domain\_pub.pid AND domain.name LIKE 'Databases'; \\
        
        \hline 6 & SELECT org.name FROM writes, pub, domain\_conf, domain, author, org, conf, domain\_pub WHERE writes.pid = pub.pid AND writes.aid = author.aid AND pub.pid = domain\_pub.pid AND pub.cid = conf.cid AND domain\_conf.did = domain.did AND domain\_conf.cid = conf.cid AND author.oid = org.oid AND pub.year > 2005 AND domain.name LIKE 'Databases'; \\
        
        \hline 7 & SELECT author.name FROM writes, pub, conf, author, org WHERE author.oid = org.oid AND writes.pid = pub.pid AND writes.aid = author.aid AND pub.cid = conf.cid AND conf.name LIKE 'VLDB' AND org.name LIKE `Tel Aviv University'; \\
        
        \hline 8 & SELECT conf.name FROM org, author, writes, conf, pub, WHERE org.oid = author.oid AND author.aid = writes.aid AND writes.pid = pub.pid AND conf.cid = pub.cid AND pub.year = 2005; \\
        
        \hline 9 & SELECT pub.year FROM author, writes, org, pub WHERE author.aid = writes.aid AND author.oid = org.oid AND writes.pid = pub.pid AND org.name LIKE `IBM'; \\
        \hline
    \end{tabularx}
    \vspace{-4mm}
\end{table}


\begin{figure}
    \centering
    \includegraphics[trim=0cm 0.5cm 0cm 0cm, width=0.8\linewidth]{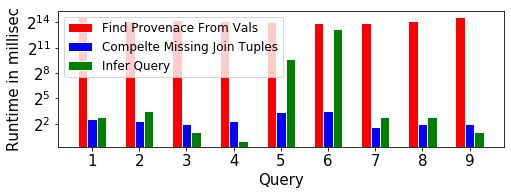}
    \caption{Runtime breakdown for query inference from \ex\ with two rows. The red bars are the time to extract the provenance tuples from the given values, the blue bars are the time to get the missing join tuples, and the green bars are the time to infer the query from the completed prov.}
    \label{fig:runtimes}
    \vspace{-3mm}
\end{figure}

We have implemented our algorithm in Python 3.6 with an underlying database engine in MySQL 5.7, and performed a preliminary evaluation of our approach. Particularly we assessed the queries inferred by the algorithm using joinless provenance, and measured the runtimes for different queries. For the first stage for inferring tuples from values we used the SQL `LIKE' function for strings. 
We evaluated queries 1--9 (Table \ref{tbl:queries} depicts three of the queries) with the MAS dataset~\cite{mas}. The examples consisted of two output examples and their provenance values. 


We observed that the inferred queries contained additional filters that were not included in the original query but are present in the \ex. The reason for that is the small number of examples given as input. This observation raises the need for interactivity in the system, that presents optional selection criteria and allows the user select or refine them based on the desired actual query. Apart from these additional filters, the algorithm inferred all queries correctly from a set of two rows. This is due to its greedy approach of trying the smaller join graphs first the returning the first query that has a set of homomorphisms that covers all provenance graphs.

The runtimes of the algorithm are shown in log-scale in Figure~\ref{fig:runtimes}. The stage that was most costly was inferring the initial provenance tuples from the database. The next stage was finding the connecting join tuples. Here the queries are more focused so this stage is faster. The last stage of inferring the query from the provenance using our homomorphism algorithm is usually very fast. However, the more complex query 6 (shown in Table \ref{tbl:queries}) contains 7 joins took $9$ seconds to infer. Since more joins are involved, the algorithm had to inspect more complex homomorphisms and, thus, took longer to find the query join graph and check its consistency. 

\section{Open problems and future work}\label{sec:conc}
We have presented an initial solution for the problem of inferring queries based on output examples with partial explanations represented as values from the database. The problem requires extensive investigation, which is the subject of our on-going work. 
In particular we consider other, more relaxed variants of explanations. 
This problem calls for different techniques that are prevalent in the field of query-by-output where column mappings are performed to match an attribute in the output to an attribute in one of the database tables. We also consider a natural language based implementation for the explanations. 


\bibliographystyle{abbrv}
\bibliography{bibtex}

\begin{thebibliography}{10}

\bibitem{Bonifati}
A.~Bonifati, R.~Ciucanu, and S.~Staworko.
\newblock Interactive join query inference with jim.
\newblock {\em PVLDB}, 7(13), 2014.

\bibitem{BonifatiCS16}
A.~Bonifati, R.~Ciucanu, and S.~Staworko.
\newblock Learning join queries from user examples.
\newblock {\em {ACM} Trans. Database Syst.}, 40(4), 2016.

\bibitem{Why}
P.~Buneman, S.~Khanna, and W.~Tan.
\newblock Why and where: A characterization of data provenance.
\newblock In {\em ICDT}, 2001.

\bibitem{lin}
Y.~Cui, J.~Widom, and J.~L. Wiener.
\newblock Tracing the lineage of view data in a warehousing environment.
\newblock {\em ACM Trans. Database Syst.}, 2000.

\bibitem{DeutchG16}
D.~Deutch and A.~Gilad.
\newblock Qplain: Query by explanation.
\newblock In {\em ICDE}, 2016.

\bibitem{DeutchG19}
D.~Deutch and A.~Gilad.
\newblock Reverse-engineering conjunctive queries from provenance examples.
\newblock In {\em EDBT}, 2019.

\bibitem{GKT-pods07}
T.~J. Green, G.~Karvounarakis, and V.~Tannen.
\newblock Provenance semirings.
\newblock In {\em PODS}, 2007.

\bibitem{Kalashnikov2018}
D.~V. Kalashnikov, L.~V. Lakshmanan, and D.~Srivastava.
\newblock Fastqre: Fast query reverse engineering.
\newblock In {\em SIGMOD}, 2018.

\bibitem{nalir}
F.~Li and H.~V. Jagadish.
\newblock Constructing an interactive natural language interface for relational
  databases.
\newblock {\em Proc. VLDB Endow.}, 2014.

\bibitem{mas}
MAS.
\newblock \url{http://academic.research.microsoft.com/}.

\bibitem{Psallidas}
F.~Psallidas, B.~Ding, K.~Chakrabarti, and S.~Chaudhuri.
\newblock S4: Top-k spreadsheet-style search for query discovery.
\newblock In {\em SIGMOD}, 2015.

\bibitem{shen}
Y.~Shen, K.~Chakrabarti, S.~Chaudhuri, B.~Ding, and L.~Novik.
\newblock Discovering queries based on example tuples.
\newblock In {\em SIGMOD}, 2014.

\end{thebibliography}

\end{document}